\documentclass[aps,11pt,twoside,letterpaper]{revtex4}

\usepackage{amsfonts, amssymb, amsmath, amsthm, latexsym}   
\usepackage{epic,eepic}					
\usepackage{graphicx}
\usepackage{xcolor}
\usepackage{pstricks}               	

\input{Qcircuit.tex}        			

\newtheorem{definition}{Definition}[section]
\newtheorem{proposition}[definition]{Proposition}
\newtheorem{lemma}[definition]{Lemma}

\newtheorem{theorem}{Theorem}[section]
\newtheorem{corollary}[definition]{Corollary}

\newcommand{\nc}{\newcommand}
\nc{\rnc}{\renewcommand}
\nc{\beq}{\begin{equation}}
\nc{\eeq}{{\end{equation}}}
\nc{\beqa}{\begin{eqnarray}}
\nc{\eeqa}{\end{eqnarray}} \nc{\lbar}[1]{\overline{#1}}
\nc{\ketbra}[2]{|#1\rangle\!\langle#2|}
\nc{\proj}[1]{|#1\rangle\!\langle #1 |}
\nc{\avg}[1]{\langle#1\rangle}
\rnc{\max}{\operatorname{max}}
\nc{\rank}{\operatorname{rank}}
\nc{\conv}{\operatorname{conv}}
\nc{\smfrac}[2]{\mbox{$\frac{#1}{#2}$}}
\nc{\Tr}{\operatorname{Tr}}
\nc{\ox}{\otimes}
\nc{\dg}{\dagger}
\nc{\dn}{\downarrow} \nc{\cA}{{\cal A}} \nc{\cB}{{\cal B}}
\nc{\cC}{{\cal C}} \nc{\cD}{{\cal D}} \nc{\cE}{{\cal E}}
\nc{\cF}{{\cal F}} \nc{\cG}{{\cal G}} \nc{\cH}{{\cal H}}
\nc{\cI}{{\cal I}} \nc{\cJ}{{\cal J}} \nc{\cK}{{\cal K}}
\nc{\cL}{{\cal L}} \nc{\cM}{{\cal M}} \nc{\cN}{{\cal N}}
\nc{\cO}{{\cal O}} \nc{\cP}{{\cal P}} \nc{\cR}{{\cal R}}
\nc{\cS}{{\cal S}} \nc{\cT}{{\cal T}} \nc{\cU}{{\cal U}}
\nc{\cX}{{\cal X}} \nc{\cW}{{\cal W}} \nc{\cZ}{{\cal Z}}
\nc{\csupp}{{\operatorname{csupp}}}
\nc{\qsupp}{{\operatorname{qsupp}}}
\nc{\rar}{\rightarrow} \nc{\lrar}{\longrightarrow}
\nc{\poly}{\operatorname{poly}}
\nc{\polylog}{\operatorname{polylog}}
\nc{\Lip}{\operatorname{Lip}} 

\def\>{\rangle}
\def\<{\langle}

\def\ph{\varphi}

\def\ph{\varphi}

\nc{\glneq}{{\raisebox{0.6ex}{$>$}  \hspace*{-1.8ex} \raisebox{-0.6ex}{$<$}}}
\nc{\gleq}{{\raisebox{0.6ex}{$\geq$}\hspace*{-1.8ex} \raisebox{-0.6ex}{$\leq$}}}

\nc{\RR}{{{\mathbb R}}}
\nc{\CC}{{{\mathbb C}}}
\nc{\FF}{{{\mathbb F}}}
\nc{\HH}{{{\mathbb H}}}
\nc{\NN}{{{\mathbb N}}}
\nc{\ZZ}{{{\mathbb Z}}}
\nc{\PP}{{{\mathbb P}}}
\nc{\QQ}{{{\mathbb Q}}}
\nc{\UU}{{{\mathbb U}}}
\nc{\WW}{{{\mathbb W}}}
\rnc{\SS}{{{\mathbb S}}}
\nc{\id}{{\operatorname{id}}}

\nc{\vholder}[1]{\rule{0pt}{#1}}

\nc{\ob}[1]{#1}

\def\beq{\begin{equation}}
\def\eeq{\end{equation}}

\nc{\eq}[1]{Eq.~(\ref{eq:#1})} \nc{\eqs}[2]{Eqs.~(\ref{eq:#1}) and
(\ref{eq:#2})}

\nc{\eqn}[1]{Eq.~(\ref{eqn:#1})}
\nc{\eqns}[2]{Eqs.~(\ref{eqn:#1}) and (\ref{eqn:#2})}

\nc{\region}{\cS\cW}

\def \Tr{\textup{Tr}}

\def \calH{ \ensuremath{ \cal H } }

\renewcommand{\ket}[1]{\left\vert #1 \right\rangle}
\renewcommand{\bra}[1]{\left\langle #1 \right\vert}

\newcommand{\samekb}[1]{\ket{#1}\bra{#1}}

\newcommand{\K}{  \ensuremath{\mathcal K} }
\newcommand{\Q}{  \ensuremath{\mathcal Q} }
\newcommand{\Pp}{  \ensuremath{\mathcal P} }
\newcommand{\Ww}{  \ensuremath{\mathcal W} }
\newcommand{\qq}{  \ensuremath{\vec{q}} }
\newcommand{\Kbar}{{  \ensuremath{ \bar{ {\mathcal K} } }   }}

\newcommand{\Esq}{\ensuremath{E_\text{sq}} }

\newcommand{\rhot}{{\ensuremath{\tilde{\rho}}}}
\newcommand{\sigmat}{{\ensuremath{\tilde{\sigma}}}}

\newcommand{\be}{\begin{equation}}
\newcommand{\ee}{\end{equation}}
\newcommand{\bea}{\begin{eqnarray}}
\newcommand{\eea}{\end{eqnarray}}
\newcommand{\bs}{\begin{split}}
\newcommand{\es}{\end{split}}
\renewcommand{\neg}{\mbox{-}}
\newcommand{\XX}{X_1;\ldots;X_m}
\newcommand{\XcX}{X_1,\ldots,X_m}
\newcommand{\E}{   {\cal E} }

\newcommand{\beas}{\begin{eqnarray*}}
\newcommand{\eeas}{\end{eqnarray*}}

\newcommand{\myparagraph}[1]{\noindent{\bf #1}$\quad$}

\begin{document}

    \title{{\Large Distributed Compression and Multiparty Squashed Entanglement}}


    \author{David Avis}
     \email{avis@cs.mcgill.ca}

    \author{Patrick Hayden}
     \email{patrick@cs.mcgill.ca}
     \affiliation{
        School of Computer Science,
        McGill University,
        Montreal, Quebec, H3A 2A7, Canada
        }

    \author{Ivan Savov}
     \email{ivan.savov@mail.mcgill.ca}
     \affiliation{
        Physics Department,
        McGill University,
        Montreal, Quebec, H3A 2A7, Canada
        }

    \begin{abstract}
		We study a protocol in which many parties use quantum communication to transfer
		a shared state to a receiver without communicating with each other. 
		This protocol is a multiparty version of the fully quantum Slepian-Wolf protocol 
		for two senders and arises through the repeated application of the two-sender protocol.
		We describe bounds on the achievable rate region for the distributed compression problem.
		The inner bound arises by expressing the achievable rate region for our protocol
		in terms of its vertices and extreme rays and, equivalently, in terms of facet inequalities. 
		We also prove an outer bound on all possible rates for distributed compression
		based on the multiparty squashed entanglement, a measure of multiparty entanglement.
    \end{abstract}
    

    \maketitle



\section{Introduction}

	Quantum information theory studies the interconversion
	of information resources like quantum channels, states and entanglement
	for the purpose of accomplishing communication tasks
    \cite{BBPS,DHW04,DHW05b,FQSW}.
    This approach is rendered possible by the substantial body of results characterizing quantum channels
    \cite{H98,SW97,BSST99,D03} and quantum communication resources like entanglement \cite{B96,DW05,PV07}.

    In classical information theory, distributed compression is the search for the optimal rates 
    at which two parties Alice and Bob can compress and transmit information faithfully to a third party Charlie.
    If the senders are allowed to communicate among themselves then they can obviously use the
    correlations between their sources to achieve better rates.
    The more interesting problem is to ask what rates can be achieved if no communication is allowed between
    the senders.
    The classical version of this problem was solved by Slepian and Wolf \cite{SW73}.
    The quantum version of this problem was first approached in \cite{ADHW04,HOW05} and more recently in \cite{FQSW},
    which describes the Fully Quantum Slepian-Wolf (FQSW) protocol and partially solves the distributed compression
    problem for two senders.

    In this paper we generalize the results of the FQSW protocol to a multiparty scenario where $m$ senders,
    Alice 1 through Alice $m$, send quantum information to a single receiver, Charlie. We exhibit a set of
    achievable rates as well as an outer bound on the possible rates
    based on a new measure of multiparty entanglement that generalizes squashed entanglement \cite{CW04}.
    Our protocol is optimal for input states that have zero squashed entanglement, notably separable states.

    The multiparty squashed entanglement is interesting in its own right, and we develop a number of its
    properties in the paper. (It was also found independently by Yang et al. and described in a recent
    paper \cite{multisquash}.)
    While there exist several measures for bipartite entanglement with useful properties and applications
    \cite{BBPS, HHT, Ra99, VP98},
    the theory of multiparty entanglement, despite considerable effort \cite{LSSW,DCT99,CKW00,BPRST99},
    remains comparatively undeveloped.
    Multiparty entanglement is fundamentally more complicated because it cannot be described by a single number
    even for pure states.
    We can, however, define \emph{useful} entanglement measures for particular applications, and
    the multiparty squashed entanglement seems well-suited to application in the distributed compression
    problem.

    The structure of the paper is as follows.
    In section \ref{sec:distributed-compression} we describe the quantum distributed compression problem and present our
    protocol. Our results are twofold. In Theorem~\ref{thm:THM1} we give the formula for the achievable rate region
    using this protocol and in  Theorem \ref{thm:THM2} we provide a bound on the best possible rates for any protocol.
    The proof of Theorem~\ref{thm:THM1} is in section \ref{sec:THMIproof}.
    The proof of Theorem~\ref{thm:THM2} is given in section \ref{sec:THMIIproof} but before we get to it we need to
    introduce and describe the properties of the multiparty information quantity in section
    \ref{sec:multiparty-information} and multiparty squashed entanglement in section \ref{sec:squashed-entanglement}.



    \myparagraph{Notation:}     
    We will denote quantum systems as $A,B,R$ and the corresponding Hilbert spaces $\cH^A, \cH^B, \cH^R$
    with respective dimensions $d_A,d_B,d_R$.
    We denote pure states of the system $A$ by a \emph{ket} $\ket{\ph}^A$ and the corresponding density matrices as
    $\ph^A = \proj{\ph}^A$.
    We denote by $H(A)_\rho=-\Tr\left(\rho^A\log\rho^A\right)$ the von Neumann entropy of the state $\rho^A$.
    For a bipartite state $\sigma^{AB}$ we define the conditional entropy $H(A|B)_\sigma=H(AB)_\sigma-H(B)_\sigma$
    and the mutual information $I(A;B)_\sigma=H(A)_\sigma+H(B)_\sigma-H(AB)_\sigma$.
    The trace distance between states $\sigma$ and $\rho$ is
    $\|\sigma-\rho\|_1 = \mathrm{Tr}|\sigma - \rho|$ where  $|X| = \sqrt{X^{\dagger}X}$.
    The fidelity is defined to be $F(\sigma, \rho) = \mathrm{Tr}\left(\sqrt{\sqrt{\rho}\sigma\sqrt{\rho}}\right)^2 $.
    Two states that are very similar have fidelity close to 1 whereas states with little similarity will have low
    fidelity.
    Throughout this paper, logarithms and exponents are taken base two unless otherwise specified.


\section{Multiparty Distributed Compression} \label{sec:distributed-compression}

    Distributed compression of classical information involves many parties collaboratively encoding their sources
    $X_1,X_2\cdots X_m$ and sending the information to a common receiver \cite{C75}.
    In the quantum setting, the parties are given a quantum state $\ph^{A_1A_2\cdots A_m} \in \cH^{A_1A_2\cdots A_m}$
    and are asked to individually compress their shares of the state and transfer them
    to the receiver while sending as few qubits as possible \cite{ADHW04}. The objective is to successfully
    transmit the quantum information stored in the $A$ systems, meaning any entanglement with an external reference
    system, to the receiver.
    No communication between the senders is allowed and, unlike \cite{HOW05}, in this paper there is no classical
    communication between the senders and the receiver.

    In our analysis, we work in the case where we have many copies of the input state, so that
    the goal is to send shares of
    the purification $\ket{\psi}^{A_1A_2\cdots\!A_mR}=(\ket{\ph}^{A_1A_2\cdots\!A_mR})^{\otimes n}$,
    where the $A_i$'s denote the $m$ different systems and $R$ denotes the reference system, which does not participate
    in the protocol.
    Notice that we use $A_i$ to denote both the individual system associated with state $\ph$ as well the $n$-copy
    version associated with $\psi$; the meaning should be clear from the context.
    We also use the shorthand notation $A=A_1A_2\cdots\!A_m$
    to denote all the senders.

    The objective, as we have mentioned, is for the participants
    to transfer their $R$-entanglement to a third party Charlie as illustrated in Figure \ref{fig:mpFQSWdiagram}.
    Note that any other type of correlation the $A$ systems could have with an external subsystem
    is automatically preserved in this case, which implies for example that if $\ph$ were written
    as a convex combination $\ph = \sum_i p_i \ph_i$ then a successful protocol would automatically
    send the $\ph_i$ with high fidelity on average~\cite{EntFid}.

    \begin{figure}[ht] \begin{center}
        \input{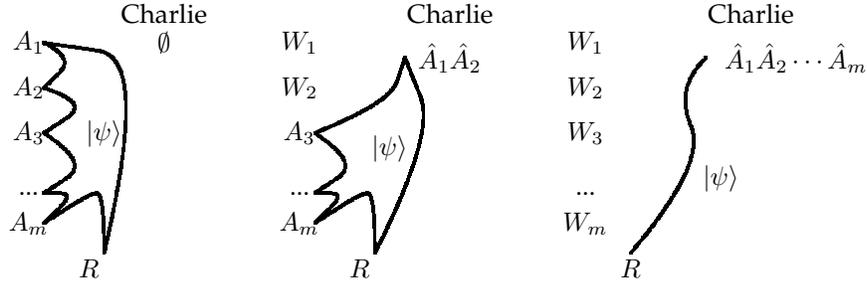}
        \caption{
            Pictorial representation of the quantum correlations between the systems at three stages of the protocol.
            Originally the state $\ket{\psi}$ is shared between $A_1A_2\cdots A_m$ and $R$.
            The middle picture shows the protocol in progress.
            Finally, all systems are received by Charlie and $\ket{\psi}$ is now shared between
            Charlie's systems $\widehat{A}_1\widehat{A}_2 \cdots \widehat{A}_m$ and $R$.            }
        \label{fig:mpFQSWdiagram}
    \end{center}
    \end{figure}

    An equivalent way of thinking about quantum distributed compression is to say that the participants
    are attempting to decouple their systems from the reference $R$ solely by sending quantum information to Charlie.
    Indeed, if we assume that originally $R$ is the purification of $A_1A_2\cdots A_m$, and
    at the end of the protocol there are no correlations between the remnant $W$ systems (see Figure
    \ref{fig:mpFQSWdiagram}) and $R$, then the purification of $R$ must have been transferred to 
    Charlie's laboratory since none of the original information was discarded.

    To perform the distributed compression task, each of the senders independently encodes her share
    before sending part of it to Charlie.
    The encoding operations are modeled by quantum operations, that is, completely positive trace-preserving (CPTP)
    maps $E_i$ with outputs $C_i$ of dimension $2^{nQ_i}$.
    Once Charlie receives the systems that were sent to him, he will apply a decoding CPTP map $D$ with output system
    $\widehat{A}=\widehat{A}_1\widehat{A}_2 \ldots \widehat{A}_m$ isomorphic to the original $A=A_1A_2\ldots A_m$.

    \begin{definition}[The rate region] \label{achievable}
    We say that a rate tuple $\vec{Q} = (Q_1,Q_2,\ldots,Q_m)$ is achievable if for all $\epsilon > 0$ there exists
    $N(\epsilon)$ such that for all $n\geq N(\epsilon)$ there exist $n$-dependent maps $(E_1,E_2,\ldots,E_m,D)$
    with domains and ranges as in the previous paragraph for which the fidelity
    between the original state, $\ket{\psi}^{A^nR^n}=\left(\ket{\ph}^{A_1A_2\cdots\!A_mR}\right)^{\otimes n}$,
    and the final state,   ${\sigma}^{\widehat{A}_1\widehat{A}_2 \ldots \widehat{A}_mR}={\sigma}^{\widehat{A}^nR^n}$,
    satisfies
    \begin{equation} \label{rate-region}
        F\left( \ket{\psi}^{A^nR^n}\!,\ {\sigma}^{\widehat{A}^nR^n} \right)=
        \phantom{.}^{\widehat{A}^nR^n}\!\!\bra{\psi}
            (D \circ (E_1 \ox \cdots \ox E_{m}))(\psi^{A^nR^n})
            \ket{\psi}^{\widehat{A}^nR^n} \geq 1 - \epsilon.
    \end{equation}
    We call the closure of the set of achievable rate tuples the rate region.
    \end{definition}

    At this point it is illustrative to review the results of the two-party state transfer protocol \cite{FQSW},
    which form a key building block for the multiparty distributed compression protocol presented in
    section \ref{subsec:protocol}.


\subsection{The FQSW protocol} \label{subsec:FQSW-protocol}
    The fully quantum Slepian-Wolf protocol \cite{FQSW} describes a procedure for simultaneous quantum state transfer
    and  entanglement distillation. This communication task can be used as a building block for nearly all the other
    protocols of quantum information theory \cite{DHW04}, yet despite its powerful applications it is fairly simple
    to implement.

    Consider a setup where the state $\ket{\psi}^{ABR}=\left(\ket{\ph}^{ABR}\right)^{\otimes n}$ is shared between
    Alice, Bob and a reference system $R$.
    The FQSW protocol describes a procedure for Alice to transfer her $R$-entanglement to Bob while
    at the same time generating ebits with him.
    Alice can accomplish this by encoding and sending part of her system, denoted $A_1$, to Bob.
    The state after the protocol can approximately be written as
    $\ket{\Phi}^{A_2\widetilde{B}}(\ket{\ph}^{R\widehat{B}})^{\otimes n}$,
    where the systems $\widetilde{B}$ and $\widehat{B}$ are held in Bob's lab while $A_2$ remains with Alice.
    The state $\ket{\Phi}^{A_2\widetilde{B}}$ is a maximally entangled state shared between Alice and Bob,
    a handy side-product which can be used to build more advanced protocols \cite{DH06,DY06}.
    Figure \ref{fig:FQSW} illustrates the entanglement structure before and after the protocol.
    \begin{figure}[b]   \begin{center}
        \input{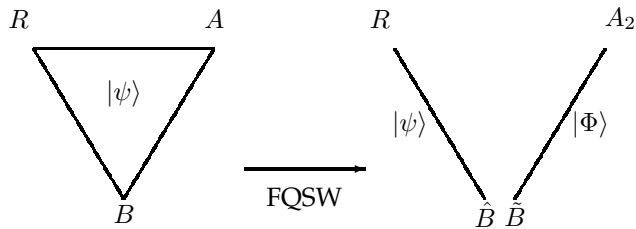}             \end{center}
        \caption{ \footnotesize
            Diagram representing the $ABR$ correlations before and after the FQSW protocol.
            Alice manages to decouple completely from the reference $R$.
            The $\widehat{B}$ system is isomorphic to $AB$. }
        \label{fig:FQSW}
    \end{figure}

    The protocol, represented graphically in Figure \ref{fig:AliceActions}, consists of the following steps:
    \begin{enumerate}
        \item   Alice performs Schumacher compression on her system $A$ to obtain the output system $A^S$.

        \item   Alice then applies a random unitary $U_A$ to $A^S$.

        \item   Next, she splits her system into two parts: $A_1A_2=A^S$ with $d_{A_1} = 2^{nQ_A}$ and
                \be
                    Q_A  >  \frac{1}{2}I(A;R)_\ph.
                \ee
                She sends the system $A_1$ to Bob.

        \item   Bob, in turn, performs a decoding operation $V_B^{ {A_1B} \to \widehat{B}\widetilde{B} }$ which splits
                his system into a $\widehat{B}$ part purifying $R$ and a $\widetilde{B}$ part which is fully entangled
                with Alice.
    \end{enumerate}

    \begin{figure}[ht]  \begin{center}
        \input{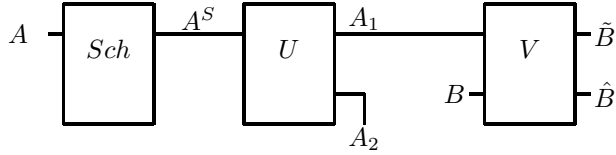} \end{center}
        \caption{
            A circuit diagram that shows the Schumacher compression and unitary encoding done by Alice and
            the decoding done by Bob.   }
        \label{fig:AliceActions}
    \end{figure}

    The best way to understand the mechanism behind this protocol is by thinking about destroying correlations.
    If, at the end of the protocol, Alice's system $A_2$ is nearly decoupled from the reference in the sense that
    $\sigma^{A_2 R} \approx \sigma^{A_2}\otimes \sigma^{R}$, then Alice must have succeeded in
    sending her $R$ entanglement to Bob because it is Bob alone who then holds the $R$ purification.
    We can therefore guess the lower bound on how many qubits Alice will have to send before she can
    decouple from the reference. Originally, Alice and R share $I(A;R)_\ph$ bits of information per copy of
    $\ket{\ph}^{ABR}$.   Since one qubit can carry away at most two bits of quantum mutual information, 
    this means that the minimum rate at which Alice must send qubits to Bob is
    \be
        Q_A     >  \frac{1}{2}I(A;R)_\ph. \label{eqn:FQSW}
    \ee

    It is shown in \cite{FQSW} that this rate is achievable in the limit of many copies of the state.
    Therefore the FQSW protocol is optimal for the state transfer task.

\subsection{The multiparty FQSW protocol} \label{subsec:protocol}
    Like the original FQSW protocol, the multiparty version relies on Schumacher compression and the mixing
    effect of random unitary operations for the encoding.
    The only additional ingredient is an agreed upon permutation of the participants.
    The temporal order in which the participants will perform their encoding is of no importance.
    However, the permutation determines how much information each participant is to send to Charlie. \\

    \noindent For each permutation $\pi$ of the participants, the protocol consists of the following steps:
    \begin{enumerate}
        \item   Each Alice-$i$ performs Schumacher compression on her system $A_i$ reducing its
                effective size to the entropy bound of roughly $H(A_i)$ qubits per copy of the state.
        \item   Each participant applies a known, pre-selected random unitary to the compressed system.
        \item   Participant $i$ sends to Charlie a system $C_i$ of dimension $2^{nQ_i}$ where
                \be
                    Q_i > \frac{1}{2}I(A_i;A_{\K_i}R)_\ph
                \ee
                where $\K_i = \{ \pi\!(j) : j>\pi^{\neg 1}(i) \}$ is the set of participants who come
                after participant $i$ according to the permutation.
        \item   Charlie applies a decoding operation $D$ consisting of the composition of the decoding maps
                $D_{\pi\!(m)} \circ \cdots \circ D_{\pi\!(2)} \circ D_{\pi\!(1)}$ defined by the individual FQSW steps
                in order to recover $\sigma^{\widehat{A}_1\widehat{A}_2 \ldots \widehat{A}_m}$
                nearly identical to the original $\psi^{A_1A_2\cdots\!A_m}$ and purifying $R$.
    \end{enumerate}


\subsection{Statement of Results} \label{subsec:statemenet-of-results}

    This subsection contains our two main theorems about multiparty distributed compression.
    In Theorem \ref{thm:THM1} we give the formula for the set of achievable rates using the multiparty FQSW protocol
    (sufficient conditions).
    Then, in Theorem \ref{thm:THM2} we specify another set of inequalities for the rates $Q_i$ which must be true
    for any distributed compression protocol (necessary conditions).
    In what follows, we consistently use $\K \subseteq \{ 1,2,\ldots m \}$ to denote any subset of the senders in
    the protocol.


    \begin{theorem}     \label{thm:THM1}
    Let $\ket{\ph}^{A_1A_2\cdots A_mR}$ be a pure state. If the inequality
        \be \label{inner-bound}
            \sum_{k\in \K} Q_k \geq  \frac{1}{2} \left[ \sum_{k\in \K}\!\left[H(A_k)_\ph\right]  +  H(R)_\ph -  H(RA_{\K})_\ph \right]
        \ee
    holds for all $\K \subseteq \{1,2,\ldots,m\}$, then the rate tuple $(Q_1,Q_2,\cdots,Q_m)$ is achievable
    for distributed compression of the $A_i$ systems.
    \end{theorem}

    Because Theorem \ref{thm:THM1} expresses a set of sufficient conditions for the protocol to succeed, we say
    that these rates are contained in the rate region. The proof is given in the next section.

    In the $m$-dimensional space of rate tuples $(Q_1,Q_2,\cdots,Q_m) \in \RR^m$, the inequalities
    \eqref{inner-bound} define a convex polyhedron \cite{poly} whose facets are given by the corresponding
    hyperplanes, as illustrated in Figure~\ref{fig:graph3d}.

    \begin{figure}[ht]
        \begin{center}      \includegraphics[width=3.4in]{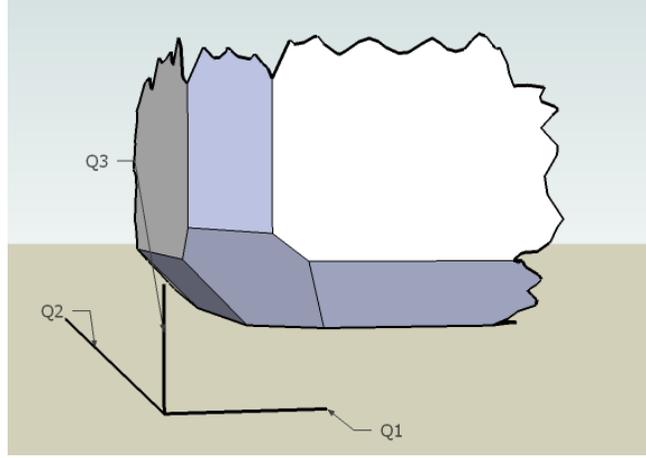}     \end{center}
        \caption{ Sketch of the rate region for the multiparty FQSW protocol for three senders. 
                                }
        \label{fig:graph3d}
    \end{figure}

    In order to characterize the rate region further we also derive an outer bound which all rate tuples must satisfy.
    \begin{theorem}     \label{thm:THM2}
    Let $\ket{\ph}^{A_1A_2\cdots A_mR}$ be a pure state input to a distributed compression protocol which
    achieves the rate tuple $(Q_1,Q_2,\ldots,Q_m)$, then it must be true that
        \be \label{outer-bound}
            \sum_{k\in \K} Q_k
            \geq
            \frac{1}{2} \left[ \sum_{k\in \K}\!\left[H(A_k)_\ph\right]  +  H(R)_\ph -  H(RA_{\K})_\ph \right]
                        - \Esq( A_{k_1};A_{k_2};\ldots ;A_{k_{|\K|}} )_\ph,
        \ee
    for all $\K \subseteq \{1,2,\ldots,m\}$, where \Esq is the multiparty squashed entanglement.
    \end{theorem}

    The multiparty squashed entanglement, independently discovered in \cite{multisquash}, is a measure of multipartite
    entanglement which generalizes the bipartite squashed entanglement of \cite{CW04}.
    Sections \ref{sec:multiparty-information} and
    \ref{sec:squashed-entanglement} below define the quantity and investigate some of its properties.
    The proof of \ref{thm:THM2} is given in section \ref{sec:THMIIproof}.

    Notice that Theorems \ref{thm:THM1} and \ref{thm:THM2} both provide bounds of the same form and
    only differ by the presence of the \Esq term.
    The rate region is squeezed somewhere between these two bounds as illustrated in Figure \ref{fig:2Dregion}.
    \begin{figure}[ht]  \begin{center}
        \def\JPicScale{0.8}
        \input{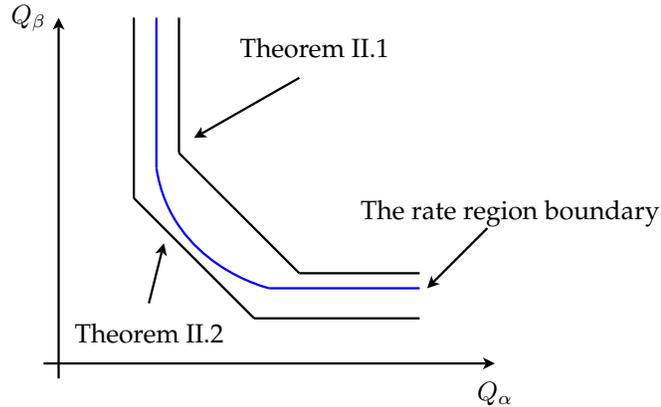} \def\JPicScale{1} \end{center}
        \caption{
            A two dimensional diagram showing the inner bound from Theorem \ref{thm:THM1} and the
            outer bound from Theorem \ref{thm:THM2}. The boundary of the real rate region must lie somewhere in between.}
        \label{fig:2Dregion}
    \end{figure}
    For states which have zero squashed entanglement, the inner and outer bounds on the region coincide so that
    in those cases our protocol is an optimal solution to the multiparty distributed compression problem.

    One can verify that when only two parties are involved ($m=2$), the inequalities in (\ref{inner-bound})
    reduce to the 2-party bounds in the original FQSW paper:
    \begin{align} \label{eqn:region}
            Q_1       &\geq \frac{1}{2} I({A_1};R), \nonumber \\
            Q_2       &\geq \frac{1}{2} I(A_2;R), \nonumber \\
            Q_1 + Q_2 &\geq \frac{1}{2}\left[ H({A_1}) + H({A_2})+ H({A_1}{A_2}) \right]. \nonumber
    \end{align}

    The family of inequalities (\ref{outer-bound}) similarly reduce to the corresponding expressions in \cite{FQSW} with
    the multiparty squashed entanglement being replaced by the original two-party squashed entanglement of \cite{CW04}.


\section{Proof of the Achievable Rates} \label{sec:THMIproof}

    %
    The multiparty fully quantum Slepian-Wolf protocol can be constructed through the repeated application of the
    two-party FQSW protocol \cite{FQSW}.
    In the multiparty case, however, the geometry of the rate region is more involved and some concepts
    from the theory of polyhedra \cite{poly} prove helpful in giving it a precise characterization.
    Multiparty rate regions in information theory have previously appeared in \cite{C75,TH98}.

    For every permutation $\pi \in S_m$ of the $m$ senders, there is a different rate tuple
    $\vec{q}_\pi = (Q_1,Q_2,\ldots,Q_m)_\pi \in \RR^m$ which is achievable in the limit of many copies of the state.
    By time-sharing we can achieve any rate that lies in the \emph{convex hull} of these points.
    We will show that the rate region for an input state $\ket{\ph}^{A_1\cdots A_mR}$ can equivalently be described
    by the set of inequalities from Theorem \ref{thm:THM1}, that is
    \be \label{inner-bound-biss}
        \sum_{k\in \K} Q_k
        \geq
        \frac{1}{2} \left[ \sum_{k\in \K}\!H(A_k)_\ph  +  H(R)_\ph -  H(RA_{\K})_\ph \right] =: C_\K
    \ee
    where $\K \subseteq \{1,2,\ldots,m\}$ ranges over all subsets of participants and $C_\K$ is the name we give to
    the constant on the right hand side of the inequality.
    The proof of Theorem \ref{thm:THM1} proceeds in two steps. First we show the set of rate tuples $\{\vec{q}_\pi\}$ is
    contained in the rate region and then we prove that the set of inequalities \eqref{inner-bound-biss} is an
    equivalent description of the rates obtained by time sharing and resource wasting of the rates $\{\vec{q}_\pi\}$.


    Consider the $m$-dimensional space of rate tuples $(Q_1,\cdots,Q_m) \in \RR^m$.
    We begin by a formal definition of a corner point $\qq_\pi$.
    \begin{definition}[Corner point]
        Let $\pi \in S_m$ be a permutations of the senders in the protocol.
        The corresponding rate tuple $q_\pi=(Q_1,Q_2,\ldots,Q_m)$ is a corner point if
        \be     \label{individual-rates}
            Q_{\pi\!(k)}    =   \frac{1}{2}I(A_{\pi\!(k)};A_{\pi\!(k+1)}\cdots A_{\pi\!(m)}R)
        \ee
        where the set $A_{\pi\!(k+1)}\cdots A_{\pi\!(m)}$ denotes all the systems which come after $k$ in
        the permutation $\pi$.
    \end{definition}

    We define $\Q := \{ \qq_\pi : \pi \in S_m \}$, the set of all corner points. 
    Clearly, $\vert \Q \vert \leq m!$ but since some permutations might lead to the same rate tuple, 
    the inequality may be strict.

    \begin{lemma}   \label{q-point}
        The set of corner points, $\Q = \{ \qq_\pi : \pi \in S_m \}$, is contained in the rate region.
    \end{lemma}

    \begin{proof}[Proof sketch for Lemma \ref{q-point}]
        We will now exhibit a protocol that achieves one such point. In order to simplify the notation,
        but without loss of generality, we choose the reversed-order permutation $\pi=(m, \ldots, 2,1)$.
        This choice of permutation corresponds to Alice-$m$ sending her information first and Alice-$1$ sending last.

        We will repeatedly use the FQSW protocol is order to send the $m$ systems to Charlie:
        \begin{enumerate}
            \item       The first party Schumacher compresses her system $A_m$ and sends it to Charlie.
                        She succeeds provided
                        \be \nonumber
                            Q_m     \geq    \frac{1}{2}I(A_m;A_1A_2\ldots A_{m-1}R) + \delta= H(A_m) + \delta
                        \ee
                        for any $\delta > 0$.
                        The above rate is dictated by the FQSW inequality \eqref{eqn:FQSW} because
                        we are facing the same type of problem except that the ``reference'' consists
                        of $R$ as well as the remaining participants $A_1A_2\cdots A_{m-1}$.
                        The fact that the formula reduces to $Q_m > H(A_m)$ should also be expected
                        since there are no correlations that the first participant can take advantage of;
                        she is just performing Schumacher compression.

            \item       The second party also faces an instance of an FQSW problem.
                        The task is to transmit the system $A_{m-1}$ to Charlie, who is now assumed to hold $A_m$.
                        The purifying system consists of  $A_1A_2\cdots A_{m-2}R$.
                        According to inequality (\ref{eqn:FQSW}) the rate must be
                        \be \nonumber
                            Q_{m-1} \geq    \frac{1}{2}I(A_{m-1};A_1A_2\cdots A_{m-2}R) + \delta
                        \ee
                        for any $\delta > 0$.

            \item       The last person to be merging with Charlie will have a purifying system consisting
                        of only $R$. Her transfer will be successful if
                        \be \nonumber
                            Q_1     \geq    \frac{1}{2}I(A_1;R) + \delta
                        \ee
                        for any $\delta > 0$.
        \end{enumerate}

        On the receiving end of the protocol, Charlie will apply the decoding map $D$ consisting of the
        composition of the decoding maps $D_1 \circ D_2 \circ \cdots \circ D_m$ defined by the individual
        FQSW steps to recover the
        state $\sigma^{\widehat{A}_1\widehat{A}_2 \cdots \widehat{A}_m}$, which will be such that
        the fidelity between $\ket{\psi}^{A^n R^n}$ and $\sigma^{\hat{A}^n R^n}$ is high, essentially by the triangle
        inequality.
        Finally, because we can make $\delta$ arbitrarily small, the rate tuple $(Q_1,\cdots,Q_m)$, with
        \be
            Q_k     =   \frac{1}{2}I(A_k; A_{1}\cdots A_{k\!-\!1}R),
        \ee
        must be contained in the rate region. The same argument applies for each permutation $\pi \in S_m$, leading
        to the conclusion that the full set $\Q$ is contained in the rate region.
    \end{proof}

    Each one of the corner points $\qq_\pi$ can also be described by an equivalent set of equations involving
    sums of the rates.
    \begin{lemma} \label{lemmaWithL}
        The rate tuple $(Q_1,Q_2,\ldots,Q_m)$ is a corner point if and only if for some $\pi \in S_m$ and for
        all $l$ such that $1\leq l \leq m$,
        \be \label{equationWithL}
            \sum_{m-l+1 \leq k \leq m} Q_{\pi\!(k)} =    \frac{1}{2}
                                                        \left[
                                                            \sum_{m-l+1 \leq k \leq m}\!\!\!\!\!\!\!\!\!
                                                            H(A_{\pi\!(k)}) + H(R)
                                                            - H(A_{\pi\![m-l+1, m]}R)
                                                        \right]
                                                        = C_{\pi\![m-l+1, m]}
        \ee
        where $A_{\pi\![m-l+1, m]} := A_{\pi\!(m-l+1)}A_{\pi\!(m-l+2)}\cdots A_{\pi\!(m)}$ denotes
        the last $l$ participants according to the permutation $\pi$.
    \end{lemma}
    \begin{proof}[Proof of Lemma \ref{lemmaWithL}]
        The proof follows trivially from Lemma \ref{q-point} by considering sums of the rates.
        If we again choose the permutation $\pi=(m, \ldots, 2,1)$ for simplicity, we see that the sum of the
        rates of the last $l$ participants is
        \begin{align}
            Q_1 +  \cdots + Q_l         &= \frac{1}{2}\bigg[ I(A_1;R)
                                                    + I(A_2;A_1R) +\cdots
                                                    + I(A_l;A_1 \cdots A_{l-1}R) \bigg] \nonumber \\
                            &=  \frac{1}{2}\bigg[
                                                        \sum_{1 \leq k \leq l} H(A_k)  + H(R) -  H(A_1\cdots A_lR)
                                                        \bigg] = C_{12\ldots l}.
        \end{align}
        A telescoping effect occurs and most of the inner terms cancel so we are left with a system of equations
        identical to \eqref{equationWithL}. Moreover, this system is clearly solvable for
        the individual rates $Q_k$. The analogous simplification occurs for all other permutations.
    \end{proof}


    So far, we have shown that the set
    of corner points $\Q$ is contained in the rate region of the multiparty fully quantum Slepian-Wolf protocol.
    The convex hull of a set of points $\Q$ is defined to be
    \be \label{conv-def}
        conv(\Q):= \left\{  \vec{x}  \in \RR^m :\
                            \vec{x} = \sum \lambda_i \vec{q}_i, \
                            \vec{q}_i \in \Q,\
                            \lambda_i \geq 0,\
                            \sum \lambda_i = 1
                    \right\}.
    \ee
    Because of the possibility of time-sharing between the different corner points, the entire convex hull
    $conv(\Q)$ must be achievable.
    Furthermore, by simply allowing any one of the senders to waste resources, we know that if a rate tuple
    $\qq$ is achievable, then so is $\qq + \vec{w}$ for any vector $\vec{w}$ with nonnegative coefficients.
    More formally, we say that any $\qq + cone(\vec{e}_1, \vec{e}_2, \ldots, \vec{e}_m)$ is also inside the
    rate region, where $\{\vec{e}_i\}$ is the standard basis for $\RR^m$:
    $\vec{e}_i = (\underbrace{0, 0, \ldots, 0, 1}_{i}, 0, 0)$ and
    \vspace{-0.4cm}
    \be \label{cone-def}
        cone(\vec{e}_1,\cdots,\vec{e}_m) :=
                \left\{     \vec{x}  \in \RR^m : \
                            \vec{x} = \sum \lambda_i \vec{e}_i,\
                            \lambda_i \geq 0 \right\}.
    \ee
    Thus, we have demonstrated that the set of rates
    \be
        P_\mathcal{V} := conv(\Q) + cone(\vec{e}_1,\cdots,\vec{e}_m)
    \ee
    is achievable. 
    To complete the proof of Theorem \ref{thm:THM1}, we will need to show that $P_\mathcal{V}$ has an equivalent
    description as
    \be \label{inner-bound-bisss}
        P_\mathcal{H}:=     \left\{ (Q_1,\cdots,Q_m) \in \RR^m \ :\
                                        \sum_{k\in \K} Q_k  \geq  C_\K ,
                                        \forall \K \subseteq \{1,2,\ldots,m\}               \right\},
    \ee
    where the constants $C_\K$ are as defined in equation \eqref{inner-bound-biss}.
    This equivalence is an explicit special case of the Minkowski-Weyl Theorem on convex polyhedra.

    \begin{theorem}[Minkowski-Weyl Theorem] \cite[p.30]{poly}
    For a subset $P \subseteq \RR^m$, the following two statements are equivalent:
    \begin{itemize}
        \item   $P$ is a $\mathcal{V}$-polyhedron: the sum of a convex hull of a finite set of points
                $\Pp = \{ \vec{p}_i\}$ plus a conical combination of vectors $\Ww = \{\vec{w}_i\}$
                \be \label{vpolyhedron}
                    P = conv(\Pp) \ +\  cone(\Ww)
                \ee
                where $conv(\Pp)$ and $cone(\Ww)$ are defined in \eqref{conv-def} and \eqref{cone-def} respectively.


        \item   $P$ is a $\mathcal{H}$-polyhedron: an intersection of $n$ closed halfspaces
                \be \label{hpolyhedron}
                    P = \{ \vec{x} \in \RR^m : A\vec{x} \geq \vec{a} \}
                \ee
                for some matrix $A \in \RR^{n\times m}$ and some vector $\vec{a} \in \RR^n$.
                Each of the $n$ rows in equation \eqref{hpolyhedron} defines one halfspace.
    \end{itemize}
    \end{theorem}

    \vspace{0.45cm}

    \myparagraph{Preliminaries}
    Before we begin the equivalence proof in earnest, we make two useful observations which will be
    instrumental to our subsequent argument.
    First, we prove a very important property of the constants $C_{\K}$ which will dictate
    the geometry of the rate region.
    \begin{lemma}[Superadditivity] \label{lemma-intersection-union}
        Let $\K, \cL \subseteq \{1,2,\ldots,m\}$ be any two subsets of the senders. Then
        \be \label{lower-order}
            C_{\K\cup\cL} + C_{\K\cap\cL} \  \geq\  C_{\K} + C_{\cL}.
        \ee
    \end{lemma}
    \begin{proof}[Proof of Lemma \ref{lemma-intersection-union}]
        We expand the $C$ terms and cancel the $\frac{1}{2}$-factors to obtain
        \begin{align} \nonumber
        \begin{aligned}
            \sum_{k\in {\K\cup\cL} }&\!H(A_k)  +  H(R) -  H(RA_{\K\cup\cL}) \\[-2mm]
            &+  \sum_{k\in {\K\cap\cL} }\!H(A_k)  +  H(R) -  H(RA_{\K\cap\cL})
        \end{aligned}
                            \ \ \ \ \ \ \ \ \ &\geq \ \ \
        \begin{aligned}
            &\sum_{k\in \K}\!H(A_k)  +  H(R) -  H(RA_{\K}) \\[-2mm]
            &\qquad  + \sum_{k\in \cL}\!H(A_k)  +  H(R) -  H(RA_{\cL}).
        \end{aligned}
        \end{align}
        After canceling all common terms we find that the above inequality is equivalent to
        %
        \be                                                                                   %
            H(RA_{\K}) +  H(RA_{\cL})
                  \ \ \ \geq \ \ \
                                H(RA_{\K\cup\cL}) + H(RA_{\K\cap\cL}),
        \ee
        which is true by the strong subadditivity (SSA) inequality of quantum entropy~\cite{LR73}.
    \end{proof}

    As a consequence of this lemma, we can derive an equivalence property for the saturated inequalities.
    \begin{corollary} \label{meet-join-lemma}
        Suppose that the following two equations hold for a given point of $P_\mathcal{H}$:
        \be
            \sum_{k\in \K} Q_k = C_{\K}             \qquad \text{and} \qquad  \sum_{k\in \cL} Q_k = C_{\cL}.
        \ee
        Then the following equations must also be true:
        \be
            \sum_{k\in \K\cup\cL} Q_k = C_{\K\cup\cL} \qquad  \text{and} \qquad  \sum_{k\in \K\cap\cL} Q_k = C_{\K\cap\cL}.
        \ee
    \end{corollary}
    \begin{proof}[Proof of Corollary \ref{meet-join-lemma}]
        The proof follows from the equation
        \be \label{chained-eqn-meet-join-lemma}
            \sum_{k\in \K} Q_k + \sum_{k\in \cL} Q_k
            =           C_{\K} + C_{\cL}
            \ \leq\     C_{\K\cup\cL} + C_{\K\cap\cL}
            \ \leq\     \sum_{k\in \K\cup\cL} Q_k + \sum_{k\in \K\cap\cL} Q_k
        \ee
        where the first inequality comes from Lemma \ref{lemma-intersection-union}.
        The second inequality is true by the definition of $P_\mathcal{H}$ since $\K\cup\cL$ and $\K\cap\cL$
        are subsets of $\{1,2,\ldots,m\}$.
        Because the leftmost terms and rightmost terms are identical, we must have equality
        throughout equation \eqref{chained-eqn-meet-join-lemma}, which in turn implies the the union and the
        intersection equations are saturated.
    \end{proof}

	An important consequence of Lemma~\ref{lemma-intersection-union}  is that it implies that the polyhedron
	$P_\mathcal{H}$ has a very special structure. It is known as a supermodular polyhedron or
	contra-polymatroid. 
	The fact that $conv(Q) = P_\mathcal{H}$ was proved by Edmonds \cite{E69},
	whose ingenious proof makes use of linear programming duality. Below we give an elementary proof 
	that does not use duality.

    A \emph{vertex} is a zero-dimensional face of a polyhedron. 
	A point $ \bar{Q} = ( \bar{Q}_1,\bar{Q}_2,\ldots,\bar{Q}_m ) \in P_\mathcal{H} \subset \RR^m$
	is a vertex of $P_\mathcal{H}$ if and only if it is the unique solution of a set of linearly
	independent equations
    \be     \label{original-set-of-equations}
            \sum_{k\in \cL_i} Q_k        =  C_{\cL_i}, \qquad \qquad 1 \leq i \leq m
    \ee
    for some subsets $\cL_i \subseteq \{1,2,\ldots,m\}$. 
    In the remainder of the proof we require only a specific consequence of linear independence,
    which we state in the following lemma.
    \begin{lemma}[No co-occurrence] \label{no-co}
        Let $\cL_i \subseteq \{1,2,\ldots,m\}$ be a collection of $m$ sets such that the system
        \eqref{original-set-of-equations} has a unique solution. 
        Then there is no pair of elements $j$, $k$ such that $j \in \cL_i$ if and only if  $k \in \cL_i$  for all $i$.
    \end{lemma}
    \begin{proof}
        If there was such a pair $j$ and $k$, then the corresponding columns of the left hand side of 
        \eqref{original-set-of-equations} would be linearly dependent.
    \end{proof}

    Armed with the above tools, we will now show that there is a one-to-one correspondence between the
    corner points $\Q$ and the vertices of the $\mathcal{H}$-polyhedron $P_\mathcal{H}$.
    We will then show that the vectors that generate the cone part of the $\mathcal{H}$-polyhedron
    correspond to the resource wasting vectors $\{\vec{e}_i\}$.

    \vspace{0.45cm}

    \myparagraph{Step 1: $\Q \subseteq vertices(P_\mathcal{H})$}
    We know from Lemma \ref{lemmaWithL} that every point $\qq_\pi \in \Q$ satisfies the $m$ equations
	\begin{align}	\label{special-set-of-equations}
        \sum_{m-i+1 \leq k \leq m} Q_{\pi\!(k)} &=   C_{\pi\![m-i+1, m]}, & 1\leq &i \leq m.
	\end{align}



	The equations \eqref{special-set-of-equations} are linearly independent since the left hand
	side is triangular, and have the form of the inequalites in \eqref{inner-bound-bisss}
	that are used to define $P_\mathcal{H}$. They have the unique solution:
	\begin{align}	\label{solution}
	 Q_{\pi\!(m)} &=  C_{\pi(m)}  &	
	 Q_{\pi\!(i)} &=   C_{\pi\![i, m]} - C_{\pi\![i+1, m]}, \qquad 1\leq i \leq m-1.
	\end{align}
	We need to show that this solution satisfies all the inequalities used to define
	$P_\mathcal{H}$ in \eqref{inner-bound-bisss}. We proceed by induction on $|\K|$.
	The case $|\K|=1$ follows from \eqref{solution} and the superadditivity property
	\eqref{lower-order}.
	For $|\K| \ge 2$ we can write $\K= \{\pi(i)\} \cup \K'$ for some 
	$\K' \subseteq \{ \pi(i+1), \pi(i+2),\ldots,\pi(m) \}$. Then
	\beas
	\sum_{k \in \K} Q_k 	&=& 	Q_{\pi(i)} + \sum_{k \in \K'} Q_k \\
							&\ge&  	C_{\pi\![i, m]} - C_{\pi\![i+1, m]} +  \sum_{k \in \K'} Q_k \\
							&\ge&  	C_{\pi\![i, m]} - C_{\pi\![i+1, m]} + C_{\K'} \qquad\qquad \textrm{(induction)} \\
							&\ge& 	C_{\K}
	\eeas
	where we again used superadditivity to get the last inequality.

    \vspace{0.45cm}

    \myparagraph{Step 2: $vertices(P_\mathcal{H}) \subseteq \Q$} In order to prove the opposite inclusion,
    we will show that every vertex of $P_\mathcal{H}$ is of the form of Lemma
    \ref{lemmaWithL}. More specifically, we want to prove the following proposition.
    \begin{proposition}[Existence of a maximal chain] \label{flag-existence}
        Every vertex of $P_\mathcal{H}$, that is, the intersection of $m$ linearly independent hyperplanes
        \begin{align}
            \qquad\qquad\qquad\qquad\qquad \sum_{k\in \cL_i} Q_k         &=  C_{\cL_i},       & 1\leq &i \leq m,\\[-3mm]
        \intertext{	defined by the family of sets $\{\cL_i; \, 1 \leq i \leq m\}$ can be 
        			described by an equivalent set of equations}
            \qquad\qquad\qquad\qquad\qquad \sum_{k\in \K_i} Q_k         &=  C_{\K_i},       & 1\leq &i \leq m,
        \end{align}
        for some family of sets distinct $\K_i \subseteq \{1,2,\ldots,m\}$ that form a \emph{maximal chain} 
        in the sense of
        \be
            \emptyset = \K_0 \subset \K_1 \subset \K_2 \subset \cdots \subset \K_{m-1} \subset \K_m = \{1,2,\ldots,m\}.
        \ee
    \end{proposition}
    Since there exists a permutation $\pi$ such that $\forall i,\ \pi\![m-i+1, m] = \K_i$ this implies that
    all the vertices of $P_\mathcal{H}$ are in $\Q$. The main tool we have have at our disposal in order to prove
    this proposition is Corollary \ref{meet-join-lemma}, which we will use extensively.



    \begin{proof}[Proof of Proposition \ref{flag-existence}]

    Let $\{ \cL_i \}_{i=1}^{m}$ be the subsets of $\{1,2,\ldots,m\}$ for which the inequalities are saturated
    and define $\cL^\mathcal{S}_i := \cL_i \cap \mathcal{S}$, the intersection of $\cL_i$ with some set
    $\mathcal{S} \subseteq \{1,2,\ldots,m\}$.

    \noindent Construct the directed graph $G = (V, E)$, where:
    \begin{itemize}
        \item   $V = \{1,2,\ldots,m\}$, i.e. the vertices are the numbers from $1$ to $m$;

        \item   $E = \left\{(j, k) \ :\  (\forall i) \; j \in \cL_i \implies k \in \cL_i \ \right\}$,
                i.e. there is an edge from vertex $j$ to vertex $k$ if  whenever vertex $j$ occurs
                in the given subsets, then so does vertex $k$.
    \end{itemize}
    Now $G$ has to be acyclic by Lemma \ref{no-co}, so it has a topological sorted order. Let us call this order $\nu$.
    Let $\K_0 = \emptyset$ and let
    \be
        \K_l = \{\nu_{m-l+1}, \ldots, \nu_m \}
    \ee
    for $l \in \{ 1, \ldots , m \}$.
    The sets $\K_l$, which consist of the last $l$ vertices according to the ordering $\nu$,
    form a maximal chain $\K_0 \subset \K_1 \subset \cdots \subset \K_{m-1} \subset \K_m$
    by construction.

    We claim that all the sets $\K_l$ can be constructed from the sets $\{ \cL_i \}$ by using unions
    and intersections as dictated by Corollary \ref{meet-join-lemma}.
    The statement is true for $\K_m=\{1,2,\ldots,m\}$ because every variable must appear in some constraint
    equation, giving $\K_m = \cup_i \cL_i$.
    The statement is also true for $\K_{m-1}=\{\nu_2, ..., \nu_m\}$ since the vertex $\nu_1$ has no
    in-edges in $G$ by the definition of a topological sort, which means that
    \be \label{build-Km1}
      \K_{m-1} = \bigcup_{\nu_1 \notin \cL^{\K_m}_i} \cL^{\K_m}_i.
    \ee
    %
    %
    %
    For the induction statement, let $l \in \{m-1, \ldots, 2, 1\}$ and
    assume that $\K_l = \bigcup_{i} \cL^{\K_{l}}_i$. 
    Since the vertex $\nu_{m-l}$ has no in-edges in the induced subgraph generated by the vertices $\K_l$
    by the definition
    of the topological sort, $\K_{l-1}$ can be obtained from the union of all the sets not containing $\nu_{m-l}$:
    \be
        \K_{l-1} = \bigcup_{\nu_{m-l} \notin \cL^{\K_l}_i} \cL^{\K_l}_i.
    \ee
    In more detail, we claim that for all $\omega \neq \nu_{m-l} \in \K_{l-1}$ there exists $i$ such that
    $\nu_{m-l} \not\in \cL_i^{\K_l}$ and $\omega \in \cL_i^{\K_l}$. If it were not true, that would imply the existence
    of $\omega \neq \nu_{m-l} \in \K_{l-1}$ such that for all $i$, $\nu_{m-l} \in \cL_i^{\K_l}$ or
    $\omega \not\in \cL_i^{\K_l}$. This last condition implies that whenever $\omega \in \cL_i^{\K_l}$
    it is also true that $\nu_{m-l} \in \cL_i^{K_l}$,
    which corresponds to an edge $(\omega,\nu_{n-l})$ in the induced subgraph.
    \end{proof}

    We have shown that every vertex can be written in precisely the same form as Lemma \ref{lemmaWithL}
    and is therefore a point in $\Q$.
    This proves $vertices(P_\mathcal{H}) \subseteq \Q$, which together with the result of Step 1, implies
    $vertices(P_\mathcal{H}) = \Q$.

    \vspace{0.45cm}

    \myparagraph{Step 3: Cone Part}
    The final step is to find the set of direction vectors that correspond to the cone part of $P_\mathcal{H}$.
    The generating vectors of the cone are all vectors that satisfy the homogeneous versions of
    the halfspace inequalities \eqref{hpolyhedron}, which in our case gives
    \be
        \sum_{k \in \K} Q_k \geq 0
    \ee
    for all $\K \subset \{ 1, 2, \ldots, m \}$. These inequalities are satisfied if and only if $Q_k \geq 0$
    for all $k$. We can therefore conclude that the cone part of $P_\mathcal{H}$ is
    $cone(\vec{e}_1,\vec{e}_2,\ldots,\vec{e}_m)$.

    \vspace{0.45cm}

    This completes our demonstration that $P_\mathcal{V}$ is the $\mathcal{V}$-polyhedron description of the
    $\mathcal{H}$-polyhedron $P_\mathcal{H}$.
    Thus we arrive at the statement we were trying to prove; if the inequalities
    \be
        \sum_{k\in \K} Q_k
                    \geq    C_{\K}
                    =       \frac{1}{2} \left[ \sum_{k\in \K}\!H(A_k)_\ph  +  H(R)_\ph -  H(RA_{\K})_\ph \right]
    \ee
    are satisfied for any $\K \subseteq \{1,2,\ldots,m\}$, then the rate tuple $(Q_1,Q_2,\cdots,Q_m)$ is inside the rate
    region. This completes the proof of Theorem \ref{thm:THM1}.



An important discovery by Edmonds \cite{E69} is that optimizing a linear function over
a supermodular polyhedron can be done in an almost trivial manner by the greedy algorithm.
Indeed, let $c_1, c_2, ... , c_m$ be any given scalars, and suppose we wish to solve
the linear program:
\[
min~ \sum_{i=1}^m c_i Q_i~~~~~~~~ \textrm{for} ~~~ (Q_1, Q_2 , ... , Q_m ) \in  P_\mathcal{H}.
\]
Let $\pi$ be the permutation such that
\[
c_{\pi(1)} \ge c_{\pi(2)} \ge ... \ge c_{\pi(m)}. 
\]
Edmonds showed that \eqref{solution} gives an optimum solution to the above
linear program.
We note in passing that we have no idea how hard it is to optimize over the 
region described by Theorem  \ref{thm:THM2}.

\section{Multiparty Information} \label{sec:multiparty-information}
    In this section and the following we present some tools that we will need in order to prove the outer bound on the
    rate region stated in Theorem~\ref{thm:THM2}.
    The following quantity is one possible generalization of the mutual information $I(A;B)$ for multiple parties.

    \begin{definition}[Multiparty Information]  \label{Jformation}
        Given the state $\rho^{X_1X_2\ldots X_m}$ shared between $m$ systems, we define the multiparty information
        as the following quantity:
        \bea
            I(X_1;X_2; \cdots; X_m)_\rho    &:=& H(X_1) + H(X_2) + \cdots + H(X_m) - H(X_1X_2 \cdots X_m)
                                            \nonumber \\
                                            &=& \sum_{i=1}^m H(X_i) - H(X_1X_2 \cdots X_m)
        \eea
    \end{definition}

    The subadditivity inequality for quantum entropy ensures that the multiparty information
    is zero if and only if $\rho$ has the tensor product form
    $\rho^{X_1} \otimes \rho^{X_2} \otimes \cdots \otimes \rho^{X_m}$.
    The conditional version of the multiparty mutual information is obtained by replacing all the entropies by
    conditional entropies
    \bea
        I(X_1;X_2; \cdots; X_m|E)_\rho  &:=& \sum_{i=1}^m H(X_i|E) - H(X_1X_2 \cdots X_m|E) \nonumber \\
                                        &=&      \sum_{i=1}^m H(X_iE) - H(X_1X_2 \cdots X_mE) - (m-1)H(E) \nonumber \\
                                        &=&      I(X_1;X_2; \cdots; X_m;E) - \sum_{i=1}^m I(X_i;E).
    \eea
    This definition of multiparty information has appeared previously in \cite{Lindblad,RHoro,GPW05} and
    more recently in \cite{multisquash}, where many of its properties were investigated. 

%

    Next we investigate some formal properties of the multiparty information which will be
    useful in our later analysis.

    \begin{lemma}[Merging of multiparty information terms] \label{mergingJ}
        Arguments of the multiparty information can be combined by subtracting their mutual information
        \be
            I(A;B;X_1;X_2; \cdots;X_m) - I(A;B) = I(AB;X_1;X_2; \cdots ;X_m).
        \ee
    \end{lemma}
    \begin{proof} This identity is a simple calculation. It is sufficient to expand the definitions and
    cancel terms.
    \end{proof}


    Discarding a subsystem inside the conditional multiparty information cannot lead it to increase.
    This property, more than any other, justifies its use as a measure of correlation.
    \begin{lemma}[Monotonicity of conditional multiparty information] \label{cond-monotonicity}
        \be
            I(AB;X_1;\cdots X_m| E) \geq I(A;X_1;\cdots X_m| E)
        \ee
    \end{lemma}
    \begin{proof}
        This follows easily from strong subadditivity of quantum entropy (SSA).
        \begin{align*}
            I(AB&;X_1;X_2; \ldots;X_m|E) =  \\
            = &\        H(ABE) + \sum_i H(X_iE) - H(ABX_1X_2 \ldots X_mE) -mH(E)    \nonumber   \\
            = &\        H(ABE) + \sum_i H(X_iE) - H(ABX_1X_2 \ldots X_mE) -mH(E)    + \nonumber \\
              &\ \quad  \underbrace{H(AE) - H(AE)}_{=0} \quad
                        + \quad \underbrace{H(AX_1X_2 \ldots X_mE) - H(AX_1X_2 \ldots X_mE)}_{=0} \nonumber \\
            = &\        H(AE) + \sum_i H(X_iE) - H(AX_1X_2 \ldots X_m) -mH(E)  + \nonumber \\
              &\ \quad \underbrace{\left[
                                        H(ABE) + H(AX_1X_2 \ldots X_mE) - H(AE) - H(ABX_1X_2 \ldots X_mE)
                                    \right]}_{\geq 0 \; \mbox{\tiny by SSA}} \nonumber \\
            \geq &\        H(AE) + \sum_i H(X_iE) - H(AX_1X_2 \ldots X_mE) -mH(E)  \\
            = &\        I(A;X_1;X_2 \ldots X_m|E)
        \end{align*}
    \end{proof}

    We will now prove a multiparty information property that follows from a more general chain rule,
    but is all that we will need for applications.

    \begin{lemma}[Chain-type Rule] \label{useful-chain-rule}
    \be
        I(AA';\XX|E) \geq I(A;\XX|A'E)
    \ee
    \end{lemma}
    \begin{proof}
    \begin{align*}
           I(AA'&;\XX|E)        = \\
            =&\ \       H(AA'E)+ \sum_{i=1}^m H(X_iE) - H(AA'\XcX) - mH(E)  \\
            =&\ \       I(A;\XX|A'E) + \sum_{i=1}^m \left[ H(A'E) + H(X_iE)- H(E) - H(A'X_iE)  \right] \nonumber \\
         \geq&\ \       I(A;\XX|A'E).
    \end{align*}
    The inequality is true by strong subadditivity.         
    \end{proof}

    \myparagraph{Remark} It is interesting to note that we have two very similar reduction-of-systems formulas derived
    from different perspectives. From Lemma \ref{cond-monotonicity} (monotonicity of the multiparty information)
    we have that
        \be
            I(AB;\XX|E) \geq I(A;\XX|E),
        \ee
    but we also know from Lemma \ref{useful-chain-rule} (chain-type rule) that
        \be
            I(AB;\XX|E) \geq I(A;\XX|BE).
        \ee
    The two expressions are inequivalent; one is not strictly stronger than the other.
    We use both of them depending on whether we want to keep the deleted system around for conditioning.

\section{Squashed Entanglement} \label{sec:squashed-entanglement}

    Using the definition of the conditional multiparty information from the previous
    section, we can define a multiparty squashed entanglement analogous to the bipartite version
    \cite{tucci-1999,tucci-2002,CW04}.
    The multiparty squashed entanglement has been investigated independently by Yang et al. \cite{multisquash}.
    For the convenience of the readers and authors alike, we will provide full proofs of all the \Esq properties
    relevant to the distributed compression problem. \\

    \begin{definition}[Multiparty Squashed Entanglement]  \label{squashedEntanglement}
        Consider the state $\rho^{X_1X_2\ldots X_m}$ shared by $m$ parties. We define the multiparty squashed
        entanglement in the following manner:
        \bea
            \Esq(X_1;X_2; \ldots; X_m)_\rho
                        &:=&        \frac{1}{2}\inf_E
                                        \left[
                                            \sum_{i=1}^m H(X_i|E)_\rhot - H(X_1X_2 \cdots X_m|E)_\rhot
                                        \right] \nonumber \\
                        &=&             \frac{1}{2}\inf_E  I(X_1;X_2;\cdots;X_m|E)_\rhot
        \eea
        where the infimum is taken over all states $\rhot^{X_1X_2\ldots X_mE}$ such that
        $\Tr_E\!\left(\rhot^{X_1X_2\ldots X_mE}\right) = \rho^{X_1X_2\ldots X_m}$. (We say $\rhot$
        is an \emph{extension} of $\rho$.)
    \end{definition}
    The dimension of the extension system $E$ can be arbitrarily large, which is in part what makes calculations
    of the squashed entanglement very difficult except for simple systems.
    The motivation behind this definition is that we can include a copy of all classical correlations inside
    the extension $E$ and thereby eliminate them from the multiparty information by conditioning.
    Since it is impossible to copy quantum information, we know that taking the infimum over all possible
    extensions $E$ we will be left with a measure of the purely quantum correlations.




    \myparagraph{Example:} It is illustrative to calculate the squashed entanglement for separable states, which
    are probabilistic mixtures of tensor products of local pure states.
    Consider the state \vspace{-0.15cm}
    \be
        \rho^{X_1X_2\ldots X_m} = \sum_j p_j        \samekb{\alpha_j}^{X_1} \otimes
                                                    \samekb{\beta_j}^{X_2} \otimes \cdots
                                                    \samekb{\zeta_j}^{X_m},                     \nonumber
                                                    \vspace{-0.3cm}
    \ee
    which we choose to extend by adding a system $E$ containing a record of the index $j$ as follows
    \be
        \rhot^{X_1X_2\ldots X_mE} = \sum_j p_j      \samekb{\alpha_j}^{X_1} \otimes
                                                    \samekb{\beta_j}^{X_2} \otimes \cdots
                                                    \samekb{\zeta_j}^{X_m}  \otimes
                                                    \samekb{j}^E.                   \nonumber
                                                    \vspace{-0.3cm}
    \ee
    When we calculate conditional entropies we notice that for any subset $\K \subseteq \{1,2,\ldots m\}$,
    \be
        H(X_\K | E)_\rhot = 0.
    \ee
    Knowledge of the classical index leaves us with a pure product state for which all the relevant entropies are zero.
    Therefore, separable states have zero squashed entanglement:
    \be
        \Esq(X_1;X_2; \ldots ;X_m)_\rho
                    = \frac{1}{2}\left[ \sum_i^m H(X_i|E)_\rhot -H(X_1X_2 \ldots X_m|E)_\rhot \right]
                    = 0.
    \ee


    We now turn our attention to the properties of $\Esq$.
    Earlier we showed that the squashed entanglement measures purely quantum contributions
    to the mutual information between systems, in the sense that it is zero for all separable states.
    In this section we will show that the multiparty squashed entanglement cannot increase under the action of
    local operations and classical communication, that is, that $\Esq$ is an LOCC-monotone.
    We will also show that \Esq has other desirable properties; it is convex, subadditive and continuous.


    \begin{proposition}
        The quantity $\Esq$ is an entanglement monotone, i.e. it does not increase on average 
        under local quantum operations
        and classical communication (LOCC). 
    \end{proposition}

    \begin{proof}
    In order to show this we will follow the argument of \cite{CW04}, which in turn follows the approach
    described in \cite{Vid00}.  We will show that \Esq has the following two properties:
    \begin{enumerate}
        \item   Given any unilocal quantum instrument $\E_k$
                (a collection of completely positive maps such that $\sum_k\!\E_k$ is trace preserving \cite{DL70})
                and any quantum state $\rho^{X_1\ldots X_m}$, then
                \be
                    \Esq(X_1;X_2;\ldots X_m)_\rho \geq \sum_k p_k \Esq(X_1;X_2;\ldots X_m)_{\rhot_k}
                \ee
                where
                \be
                    p_k = \Tr\ \E_k(\rho^{X_1\ldots X_m})
                    \quad \textrm{and} \quad
                    \rhot_k^{X_1\ldots X_m}=\frac{1}{p_k}\E_k(\rho^{X_1\ldots X_m}).
                \ee
        \item   $\Esq$ is convex.
    \end{enumerate}


    Without loss of generality, we assume that $\E_k$ acts on the first system.
    We will implement the quantum instrument by appending to $X_1$
    environment systems $X_1'$ and $X_1''$ prepared in standard pure states,
    applying a unitary $U$ on $X_1X_1'X_1''$, and then tracing out over $X_1''$.
    We store $k$, the classical record of which $\E_k$ occurred, in the $X_1'$ system.
    More precisely, for any extension of $\rho^{X_1X_2\cdots X_m}$ to $X_1X_2\cdots X_mE$,
        \be
            \rho^{X_1X_2\ldots X_mE}    \mapsto
            \rhot^{X_1X_1'X_2\ldots X_mE} := \sum_k \ \E_k\!\! \otimes\!\!
                    I_E \!\left( \rho^{X_1X_2\ldots X_mE} \right)\otimes\ket{k}\bra{k}^{X_1'}.
        \ee
    The argument is then as follows:
    \bea
      \frac{1}{2} I(X_1;X_2;\ldots X_m|E)_\rho  &=&
                \frac{1}{2}I(X_1X_1'X_1'';X_2;\ldots ;X_m|E)_\rho \label{refone}    \\
        &=&     \frac{1}{2}I(X_1X_1'X_1'';X_2;\ldots ;X_m|E)_{\rhot}    \label{reftwo} \\
        &\geq&  \frac{1}{2}I(X_1X_1';X_2;\ldots ;X_m|E)_{\rhot} \label{refthree} \\
        &\geq&  \frac{1}{2}I(X_1;X_2;\ldots; X_m|EX_1')_{\rhot} \label{reffour} \\
        &=&     \frac{1}{2}\sum_k p_k I(X_1;X_2;\ldots ;X_m|E)_{\rhot_k}    \label{reffive} \\
        &\geq&  \sum_k p_k \Esq\left(X_1;X_2;\ldots ;X_m \right)_{\rhot_k} \label{refsix}
    \eea

    The equality \eqref{refone} is true because adding an uncorrelated ancilla does not change the entropy of the system.
    The transition $\rho \rightarrow \rhot$ is unitary and doesn't change entropic quantities so
    \eqref{reftwo} is true.
    For \eqref{refthree} we use the monotonicity of conditional multiparty information, Lemma \ref{cond-monotonicity}.
    In \eqref{reffour} we use the chain-type rule from Lemma \ref{useful-chain-rule}.
    In \eqref{reffive} we use the index information $k$ contained in $X_1'$.
    Finally, since $\Esq$ is the infimum over all extensions, it must be no more than the particular extension $E$,
    so \eqref{refsix} must be true.
    Now since the extension $E$ in \eqref{refone} was arbitrary, it follows that
    $\Esq({X_1;X_2;\ldots ;X_m})_\rho \geq \sum_k p_k \Esq\left(X_1;X_2;\ldots; X_m \right)_{\rhot_k}$ which completes
    the proof of Property 1.

    To show the convexity of $\Esq$, we again follow the same route as in \cite{CW04}.
    Consider the states $\rho^{X_1X_2\ldots X_m}$ and $\sigma^{X_1X_2\ldots X_m}$ and their
    extensions $\rhot^{X_1X_2\ldots X_mE}$ and $\sigmat^{X_1X_2\ldots X_mE}$ defined over the same system $E$.
    We can also define the weighted sum of the two states
    $\tau^{X_1X_2\ldots X_m} = \lambda\rho^{X_1X_2\ldots X_m} + (1-\lambda)\sigma^{X_1X_2\ldots X_m}$
    and the following valid extension:
    \be
        \tilde{\tau}^{X_1X_2\ldots X_mEE'} = \lambda\rho^{X_1X_2\ldots X_mE}\otimes\samekb{0}^{E'}
                                            + (1-\lambda)\sigma^{X_1X_2\ldots X_mE}\otimes\samekb{1}^{E'}.
    \ee
    Using the definition of squashed entanglement we know that
    \begin{align}
        \Esq(X_1;X_2;\ldots; X_m)_\tau
                            &\leq\      \frac{1}{2} I(X_1;X_2;\ldots; X_m|EE')_{\tilde{\tau}}   \nonumber \\
                            &=\         \frac{1}{2}\left[ \lambda I(X_1;X_2;\ldots ;X_m|E)_\rhot
                                            + (1-\lambda)I(X_1;X_2;\ldots ;X_m|E)_\sigmat \right].  \nonumber
    \end{align}
    Since the extension system $E$ is completely arbitrary we have
    \be
        \Esq(X_1;X_2;\ldots; X_m)_\tau
        \leq
        \lambda\Esq(X_1;X_2;\ldots; X_m)_\rho + (1-\lambda)\Esq(X_1;X_2;\ldots; X_m)_\sigma,
    \ee
    so \Esq is convex.

    We have shown that \Esq satisfies both Properties 1 and 2. Therefore, it must be an entanglement
    monotone.
    \end{proof}

\myparagraph{Subadditivity on Product States}
    Another desirable property for measures of entanglement is that they should be additive or at least subadditive on
    tensor products of the same state. Subadditivity of \Esq is easily shown from the properties of
    multiparty information.

    \begin{proposition}     \label{subadditiveTensorProducts}
    $\Esq$ is subadditive on tensor product states, i.e.
        \be \label{subadditivity}
            \Esq\left( \rho^{X_1Y_1;X_2Y_2;\ldots;X_mY_m} \right) \leq
                    \Esq\left( \rho^{X_1;X_2;\ldots;X_m} \right)
                    +  \Esq\left( \rho^{Y_1;Y_2;\ldots;Y_m} \right)
        \ee
        where $\rho^{X_1Y_1X_2Y_2\ldots X_mY_m}=\rho^{X_1X_2\ldots X_m}\otimes\rho^{Y_1Y_2\ldots Y_m}$.
    \end{proposition}
    \begin{proof}
    Assume that $\rho^{X_1X_2\ldots X_mE}$ and $\rho^{Y_1Y_2\ldots Y_mE'}$ are extensions. Together they form an
    extension $\rho^{X_1Y_1X_2Y_2\ldots X_mY_mEE'}$ for the product state.
    \begin{align}
        2\Esq\big(X_1Y_1;&X_2Y_2;\ldots;X_mY_m\big)_\rho \nonumber \\
            &\leq\      I(X_1Y_1;X_2Y_2;\ldots;X_mY_m|EE')  \\
            &=\         \sum_i H(X_iY_iEE') - H(X_1Y_1X_2Y_2\ldots X_mY_mEE')
                                    - (m-1)H(EE') \\
            &=\         I(X_1;X_2;\ldots;X_m|E) + I(Y_1;Y_2;\ldots;Y_m|E').
    \end{align}
    The first line holds because the extension for the $XY$ system that can be built by combining
    the $X$ and $Y$ extensions is not the most general extension.
    The proposition then follows because the inequality holds for all extensions of $\rho$ and $\sigma$.
    \end{proof}

    The question of whether \Esq is additive, meaning superadditive in addition to subadditive, remains an open problem.
    Indeed, if it were possible to show that correlation between the $X$ and $Y$ extensions is unnecessary
    in the evaluation of the squashed entanglement of $\rho \otimes \sigma$, then \Esq would be additive.
    This is provably true in the bipartite case~\cite{CW04} but the same method does not seem to work
    with three or more parties.


\myparagraph{Continuity}
    The continuity of bipartite \Esq was conjectured in \cite{CW04} and proved by Alicki and Fannes in \cite{AF04}.
    We will follow the same argument here to prove the continuity of the multiparty squashed entanglement.
    The key to the continuity proof is the following lemma which makes use of an ingenious geometric construction.

    \begin{lemma}[Continuity of conditional entropy \cite{AF04}] \label{continuityOfConditional}
    Given density matrices $\rho^{AB}$ and $\sigma^{AB}$ on the space $\calH^A \otimes \calH^B$ such that
        \be
            \| \rho - \sigma \|_1 = \frac{1}{2}\Tr | \rho - \sigma | \leq \epsilon,
        \ee
    it is true that
        \be
            \left| H(A|B)_\rho - H(A|B)_\sigma \right| \leq 4\epsilon \log d_A + 2h(\epsilon)
        \ee
    where $d_A=\dim \calH^A$ and $h(\epsilon)=-\epsilon\log\epsilon - (1-\epsilon)\log(1-\epsilon)$ is the
    binary entropy.
    \end{lemma}

    This seemingly innocuous technical lemma makes it possible to prove the continuity of \Esq in spite
    of the unbounded dimension of the extension system.
    \begin{proposition}[\Esq is continuous]
        For all $\rho^{X_1X_2\ldots X_m}$, $\sigma^{X_1X_2\ldots X_m}$  with \mbox{$\| \rho
        - \sigma \|_1 \leq \epsilon$},\\
        \mbox{$\| \Esq(\rho)-\Esq(\sigma) \| \leq \epsilon'$} where $\epsilon'$ depends on $\epsilon$ and vanishes as
        $\epsilon \rightarrow 0$.
    \end{proposition}
    The precise form of $\epsilon'$ can be found in equation \eqref{epsilonprime}.

    \begin{proof}
        Proximity in trace distance implies proximity in fidelity distance \cite{Fuchs}, in the sense that
        \be
            F(\rho^{X_1X_2\ldots X_m}, \sigma^{X_1X_2\ldots X_m})   \geq        1 -\epsilon,
        \ee
        but by Uhlmann's theorem \cite{U76} this means that we can find purifications $\ket{\rho}^{X_1X_2\ldots X_mR}$
        and  $\ket{\sigma}^{X_1X_2\ldots X_mR}$ such that
        \be
            F(\ket{\rho}^{X_1X_2\ldots X_mR}, \ket{\sigma}^{X_1X_2\ldots X_mR})     \geq        1 -\epsilon.
        \ee
        Now if we imagine some general operation $\Lambda$ that acts only on the purifying system $R$
        \bea
            \rho^{X_1X_2\ldots X_mE}    &=&     (I^{X_1X_2\ldots X_m}\otimes\Lambda^{R \rightarrow E})
                                                \samekb{\rho}^{X_1X_2\ldots X_mR}   \\
            \sigma^{X_1X_2\ldots X_mE}  &=&     (I^{X_1X_2\ldots X_m}\otimes\Lambda^{R \rightarrow E})
                                                \samekb{\sigma}^{X_1X_2\ldots X_mR}
        \eea
        we have from the monotonicity of fidelity for quantum channels that
        \be
            F({\rho}^{X_1X_2\ldots X_mE}, {\sigma}^{X_1X_2\ldots X_mE})
                \geq F(\ket{\rho}^{X_1X_2\ldots X_mR}, \ket{\sigma}^{X_1X_2\ldots X_mR})
                    \geq        1 -\epsilon,
        \ee
        which in turn implies \cite{Fuchs} that
        \be
            \| \rho^{X_1X_2\ldots X_mE} - \sigma^{X_1X_2\ldots X_mE} \|_1   \leq    2\sqrt{\epsilon}.
        \ee
        Now we can apply Lemma \ref{continuityOfConditional} to each term in the multiparty information to obtain
        \begin{align}
        \Big| I(X_1&;X_2;\ldots X_m|E)_\rho - I(X_1;X_2;\ldots X_m|E)_\sigma \Big| \nonumber \\
            &\leq\  \sum_{i=1}^m    \Big| H(X_i|E)_\rho - H(X_i|E)_\sigma \Big|
                    + \Big| H(X_1X_2\ldots X_m|E)_\rho
                    - H(X_1X_2\ldots X_m|E)_\sigma \Big| \nonumber \\
            &\leq\  \sum_{i=1}^m \left[ 8\sqrt{\epsilon} \log d_i + 2h( 2\sqrt{\epsilon}) \right]
                    + 8\sqrt{\epsilon} \log\left( \prod_{i=1}^m d_i \right)
                    + 2h( 2\sqrt{\epsilon} ) \nonumber \\
            &=\     16\sqrt{\epsilon}\log\left( \prod_{i=1}^m d_i \right)
                    + (m+1)2h( 2\sqrt{\epsilon} )
                    =: \epsilon' \label{epsilonprime}
        \end{align}
        where $d_i=\dim {\calH}^{X_i}$ and $h(.)$ is as defined in Lemma \ref{continuityOfConditional}.
        Since we have shown the above inequalities for \emph{any} extension $E$ and the quantity $\epsilon'$ vanishes
        as $\epsilon \rightarrow 0$, we have proved that \Esq is continuous.
    \end{proof}

\section{Proof of outer bound on the rate region} \label{sec:THMIIproof}    

    Armed with the new tools of multiparty information and squashed entanglement, we
    are now ready to give the proof of Theorem \ref{thm:THM2}.
    We want to show that \emph{any} distributed compression protocol which works must satisfy all of the inequalities
    of type \eqref{outer-bound} from Theorem \ref{thm:THM2}.
    We break the proof into three steps.

    \myparagraph{Step 1: Decoupling Formula}
    We know that the input system $\ket{\psi}^{A^nR^n}$ is a pure state.  
    If we account for the Stinespring dilations of each encoding and decoding operation, then we can view any protocol
    as implemented by unitary transformations with ancilla and waste.
    Therefore, the output state (including the waste systems) should also be pure.

    More specifically, the encoding operations are modeled by CPTP maps $E_i$ with outputs $C_i$ of dimension $2^{nQ_i}$.
    In our analysis we will keep the Stinespring dilations of the CPTP maps $W_i$ so the evolution as a whole will
    be unitary.
    \[
    \Qcircuit @C=1em @R=.7em {
       \lstick{A_i}      & \multigate{1}{E_i} & \qw & \rstick{C_i \quad  \leftarrow\text{to Charlie}} \qw \\
       \lstick{ \ket{0}} & \ghost{E_i}        & \qw & \rstick{W_i \quad  \! \leftarrow\text{waste}} \qw
    }
    \]

    Once Charlie receives the systems that were sent to him, he will apply a decoding CPTP map $D$ with output system
    $\widehat{A}=\widehat{A}_1\widehat{A}_2 \ldots \widehat{A}_m$ isomorphic to the original $A=A_1A_2\ldots A_m$.
    \[
    \Qcircuit @C=1em @R=.7em {
       \lstick{\bigcup_{i}^m C_i}           & \multigate{1}{D}      & \qw &
            \rstick{\widehat{A}_1\cdots\widehat{A}_m \quad \leftarrow\text{near-purification of $R$}} \qw \\
       \lstick{ \ket{0}}    & \ghost{D}   & \qw & \rstick{W_C \qquad\qquad \!\!\!\leftarrow\text{Charlie's waste}} \qw
    }
    \]

    In what follows we will use Figure \ref{fig:mpFQSW} extensively in order to keep track of the evolution and purity of
    the states at various points in the protocol.

    \begin{figure}[ht]  \begin{center}
        \input{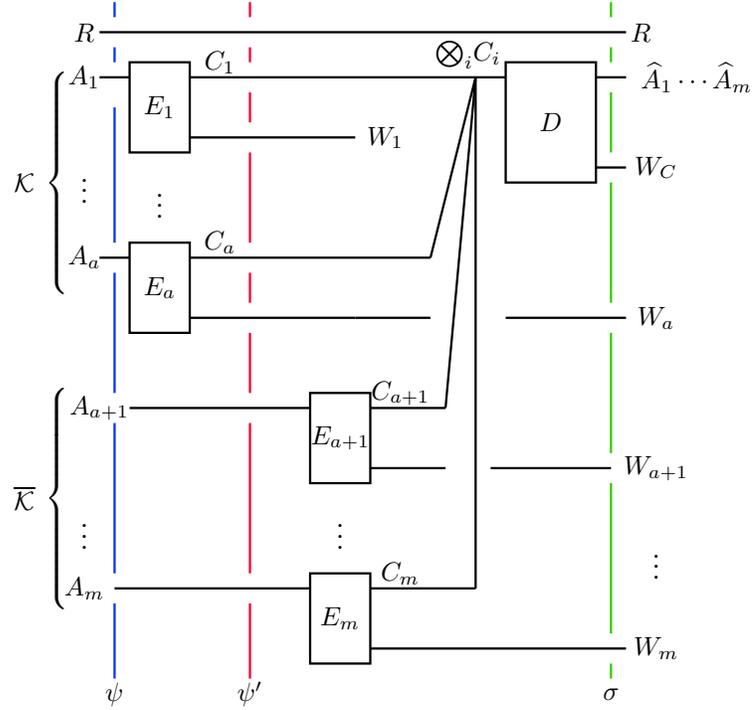} \end{center}
        \caption{
            A general distributed compression circuit diagram showing the encoding operations $E_i$ with
            output systems $C_i$ (compressed data) and $W_i$ (waste). The decoding operation takes all
            the compressed data $\bigotimes_i\!C_i$ and applies the decoding operation $D$ to output a state
            ${\sigma}^{\widehat{A}^nR^n}$ which has high fidelity with the original $\ket{\psi}^{A^nR^n}$. }
        \label{fig:mpFQSW}
    \end{figure}

    The starting point of our argument is the fidelity condition (\ref{rate-region}) for successful distributed
    compression, which we restate below for convenience
    \be
        F\left( \ket{\psi}^{A^nR^n}\!,\ {\sigma}^{\widehat{A}^nR^n} \right)     \geq    1 - \epsilon
    \ee
    where $\ket{\psi}^{A^nR^n} = \left( \ket{\ph}^{A_1A_2\cdots A_mR} \right)^{\otimes n}$ 
    is the input state to the protocol and $\sigma^{\widehat{A}^nR^n}$ 
    is the output state of the protocol.
    Since $\sigma^{\widehat{A}^nR^n}$ has high fidelity with a rank one state, it must have one large eigenvalue
    \be
        \lambda_{\rm max}(\sigma^{\widehat{A}^nR^n}) \geq 1 - \epsilon.
    \ee
    Therefore, the full output state $\ket{\sigma}^{\widehat{A}R^nW_1\!\cdots W_m W_C}$ has Schmidt decomposition
    of the form
    \be
        \ket{\sigma}^{\widehat{A}^nR^nW_1\!\cdots W_m W_C}
            =
                \sum_i \sqrt{\lambda_i} \ket{e_i}^{\widehat{A}^nR^n} \!\!\otimes \ket{f_i}^{W_1\!\cdots W_m W_C},
    \ee
    where $\ket{e_i},\ket{f_i}$ are orthonormal bases and $\lambda_1=\lambda_{\rm max} \geq 1 - \epsilon$.

    Next we show that the output state $\ket{\sigma}^{\widehat{A}^nR^nW_1\!\cdots W_m W_C}$
    is very close in fidelity to a totally decoupled state $\sigma^{\widehat{A}^nR^n}\otimes\sigma^{W_1\cdots W_mW_C}$,
    which is a tensor product of the marginals of $\ket{\sigma}$ on the subsystems
    ${\widehat{A}^nR^n}$ and ${W_1\cdots W_mW_C}$:
    \begin{align}
        F\big(\ket{\sigma}^{\widehat{A}^nR^nW_1\!\cdots W_m W_C}&, \
            \sigma^{\widehat{A}^nR^n}\otimes\sigma^{W_1\cdots W_mW_C} \big)  = \nonumber \\
            &=\ \Tr\left[ \ket{\sigma}\!\!\bra{\sigma}^{\widehat{A}^nR^nW_1\!\cdots W_m W_C}\left(
                \sigma^{\widehat{A}^nR^n}\otimes\sigma^{W_1\cdots W_mW_C}
              \right)\right] \nonumber \\
            &=\ \sum_{i}    \lambda^3_i  \geq\  (1-\epsilon)^3 \geq 1 - 3\epsilon. \label{fidelity_close}
    \end{align}
    Using the relationship between fidelity and trace distance \cite{Fuchs}, we can transform \eqref{fidelity_close}
    into the trace distance bound
    \be
        \left\| \ket{\sigma}\!\!\bra{\sigma}^{\widehat{A}^nR^nW_1\!\cdots W_m W_C} -
                    \sigma^{\widehat{A}^nR^n} \otimes \sigma^{W_1\!\cdots W_m W_C} \right\|_1 \leq 2\sqrt{3\epsilon}.
    \ee
    By the contractivity of trace distance, the same equation must be true for any subset of the systems.
    This bound combined with the Fannes inequality implies that the entropies taken with respect to
    the output state are nearly additive:
    \bea
        \big\vert H(R^n W_\K)_\sigma    \ -\  H(R^n)_\sigma + H(W_\K)_\sigma    \big\vert
        &\leq &     2\sqrt{3\epsilon} \log(d_{R^n}d_{W_\K}) + \eta(2\sqrt{3\epsilon}) \nonumber \\
        &\leq &     2\sqrt{3\epsilon} \log(d_{A^n}d_{A^{2n}_\K}) + \eta(2\sqrt{3\epsilon}) \nonumber \\
        &\leq &     2\sqrt{3\epsilon}\ n \log(d^3_{A}) + \eta(2\sqrt{3\epsilon}) = f_1(\epsilon,n).
                     \label{wastereference}
    \eea
    for any subset $\K \subseteq \{ 1,2\ldots m \}$ with $\epsilon \leq \frac{1}{12e^2}$ and $\eta(x)=-x\log x$.
    In the second line we have used the fact that $d_A = d_R$ and exploited the fact that
    $d_{W_\K}$ can be taken less than or equal to $d_{A^{2n}_\K}$, the maximum size of an environment
    required for a quantum operation with inputs and outputs of dimension no larger than $d_{A^{n}_\K}$.

    \ \\

    \myparagraph{Step 2: Dimension Counting}
    The entropy of any system is bounded above by the logarithm of its dimension.
    In the case of the systems that participants send to Charlie, this implies that
    \be \label{c-dimension-bound}
        n \sum_{k\in \K} Q_k \geq H(C_\K)_{\psi'}.
    \ee
    We can use this fact and the diagram of Figure \ref{fig:mpFQSW} to bound the rates $Q_i$.
    First we add $H(A_\Kbar)_{\psi} = H(A_\Kbar)_{\psi'}$ to both sides of equation (\ref{c-dimension-bound}) and
    obtain the inequality
    \be     \label{useful2}
        H(A_\Kbar)_{\psi}  +  n \sum_{k\in \K} Q_k
        \geq
        H(C_\K)_{\psi'} +H(A_\Kbar)_{\psi'} \geq H(C_\K A_\Kbar)_{\psi'}.
    \ee
    For each encoding operation, the input system $A_i$ is unitarily related to the outputs $C_iW_i$ so we can write
    \be   \label{thisistight}
        H(A_i)_\psi = H(W_iC_i)_{\psi'} \leq H(W_i)_{\psi'} + H(C_i)_{\psi'} \leq H(W_i)_{\psi'} + nQ_i,
    \ee
    where in the last inequality we have used the dimension bound $H(C_i) \leq nQ_i$.
    If we collect all the $Q_i$ terms from equations (\ref{useful2}) and (\ref{thisistight}), we obtain the inequalities
    \bea
        n\sum_{i\in\K} Q_i  &\geq&      H(C_\K A_\Kbar)_{\psi'} - H(A_\Kbar)_{\psi} \label{useful2rewrite} \\
        n\sum_{i\in\K} Q_i  &\geq&  \sum_{i\in\K}  H(A_i)_\psi  - \sum_{i\in\K} H(W_i)_{\psi'}. \label{sumoftight}
    \eea
    Now add equations (\ref{useful2rewrite}) and (\ref{sumoftight}) to get
    \begin{align}
    2 n\sum_{i\in\K} Q_i    &\geq^{\ \!\ \ \ }\quad \sum_{i\in\K}  H(A_i)_\psi  - \sum_{i\in\K} H(W_i)_{\psi'}
                                                +  H(C_\K A_\Kbar)_{\psi'} - H(A_\Kbar)_\psi \nonumber \\
                            &=^{(1)}\quad       \sum_{i\in\K}  H(A_i)_\psi  - \sum_{i\in\K} H(W_i)_{\psi'}
                                                +  H(W_\K R^n)_{\psi'} - H(R^n A_\K)_\psi \nonumber  \\
                            &\geq^{(2)}\quad    \sum_{i\in\K}  H(A_i)_\psi  - \sum_{i\in\K} H(W_i)_{\psi'}
                                                +  H(W_\K)_{\psi'}+ H(R^n)_{\psi'} - H(R^n A_\K)_\psi - f_1(\epsilon,n)
                                                \nonumber  \\
                            &=^{\ \ }\quad      \left[ \sum_{i\in\K} H(A_i) + H(R^n) - H(R^n A_\K) \right]_\psi
                                                +  H(W_\K)_{\psi'} - \sum_{i\in\K} H(W_i)_{\psi'} - f_1(\epsilon,n),
                                                \label{dirty-rate-bound}
    \end{align}
    where the equality $\!\!\phantom|^{(1)}$ comes about because the systems
    $\ket{\psi}^{A_\K A_\Kbar R^n}$ and $\ket{\psi'}^{C_\K W_\K A_\Kbar R^n}$ are pure.
    The inequality (\ref{wastereference}) from Step 1 was used in $\!\!\phantom|^{(2)}$.

    \ \\

    \myparagraph{Step 3: Squashed Entanglement}
    We would like to have a bound on the extra terms in equation \eqref{dirty-rate-bound} that does not depend on the
    encoding and decoding maps.
    We can accomplish this if we bound the waste terms $\sum_{i \in \K} H(W_i)_{\sigma}  - H(W_\K)_{\sigma}$
    by the squashed entanglement $2\Esq( A_{k_1};\cdots;A_{k_l} )_\psi$ of the input state
    for each $\K = \{k_1,k_2,\ldots,k_l\} \subseteq \{1,\ldots,m\}$ plus some small
    corrections.
    The proof requires a continuity statement analogous to \eqref{wastereference},
    namely that
    \be
        \big\vert H(W_i) - H(W_i|R)\big\vert \leq   f_2(\epsilon,n)     \label{conditioning-on-R}
    \ee
    where $f_2$ is some function such that $f_2(\epsilon,n)/n \rightarrow 0$ as  $\epsilon \rightarrow 0$.
    The proof is very similar to that of \eqref{wastereference} so we omit it.

    Furthermore, if we allow an arbitrary transformation  $\cE^{R \to E}$ to be applied to the $R$ system, we will
    obtain some general extension but the analog of equation (\ref{conditioning-on-R}) will remain true
    by the contractivity of the trace distance under CPTP maps. We can therefore write:
    \begin{align*} \label{wednesdayArgument}
        \sum_{i\in\K} H(W_i)_\psi  & -  H(W_\K)_\psi   \\
             &\leq     \sum_{i\in\K} H(W_i|E)  - H(W_\K|E)
                                                                     + [|\K|+1]f_2(\epsilon,n) \\
                    &=                I(W_{k_1};W_{k_2};\ldots;W_{k_l};E)
                                        - I(W_{k_1};E) - \!\!\!\!\sum_{i\in \{ \K  \setminus k_1\} } I(W_i;E) + f'_2(\epsilon,n)\\
                    &=^{(1)}      I(W_{k_1}E;W_{k_2};\ldots;W_{k_l})
                                        - \sum_{i\in \{ \K  \setminus k_1\} } I(W_i;E)  + f'_2(\epsilon,n)\\
                    &\leq^{(2)}   I(A_{k_1}E;W_{k_2};\ldots;W_{k_l})
                                        - \sum_{i\in \{ \K  \setminus k_1\} } I(W_i;E)  + f'_2(\epsilon,n)\\
                    &=^{(1)}      I(A_{k_1};W_{k_2};\ldots;W_{k_l},E) - I(A_{k_1};E)
                                        - \sum_{i\in \{ \K  \setminus k_1\} } I(W_i;E)  + f'_2(\epsilon,n)\\
                    &\leq^{(3)}       I(A_{k_1};A_{k_2};\ldots;A_{k_l};E) - \sum_{i\in\K} I(A_i;E)    + f'_2(\epsilon,n)\\
                    &\leq             I(A_{k_1};A_{k_2};\ldots;A_{k_l}|E) + f'_2(\epsilon,n),
    \end{align*}
    where we have used the shorthand $f'_2(\epsilon,n) = [|\K|+1]f_2(\epsilon,n)$ for brevity.
    Equations marked $\!\!\phantom|^{(1)}$ use Lemma \ref{mergingJ} and inequality $\!\!\phantom|^{(2)}$ comes
    about from Lemma \ref{cond-monotonicity}, the monotonicity of the multiparty information.
    Inequality $\!\!\phantom|^{(3)}$ is obtained when we repeat the steps for $k_2,\ldots,k_l$.
    The above result is true for any extension $E$ but we want to find the tightest possible lower bound for the
    rate region so we take the infimum over all possible extensions $E$ thus arriving at the definition of squashed
    entanglement.

    \ \\

    \noindent Putting together equation \eqref{dirty-rate-bound} from Step 2 and the bound from Step 3 we have
    \begin{align}
        2 n\sum_{i\in\K} Q_i    &\geq   \left[ \sum_{i\in\K} H(A_i) + H(R^n) - H(R^n A_\K) \right]_\psi
                                        - \left(
                                             \sum_{i\in\K} H(W_i)_{\psi'}-H(W_\K)_{\psi'}
                                          \right) - f_1(\epsilon,n) \nonumber \\
                                &\geq   \left[ \sum_{i\in\K} H(A_i) + H(R^n) - H(R^n A_\K) \right]_\psi
                                        -  2\Esq( A_{k_1};\cdots;A_{k_l} )_\psi
                                        - f_1(\epsilon,n) - f'_2(\epsilon,n). \nonumber
    \end{align}
    We can simplify the expression further by using the fact that $\ket{\psi} = \ket{\ph}^{\otimes n}$
    to obtain
    \begin{align}
        \sum_{k\in \K} Q_k
        &\geq   \frac{1}{2}\left[ \sum_{k\in \K} H(A_k) + H(R) - H(RA_{\K}) \right]_\ph
                        - \Esq({A_{k_1};A_{k_2};\ldots A_{k_l}})_\ph
                        - \frac{f_1(\epsilon,n)}{2n}
                        - \frac{f'_2(\epsilon,n)}{2n} \nonumber
    \end{align}
    where the we used explicitly the additivity of the entropy for tensor product states and the subadditivity of
    squashed entanglement demonstrated in Proposition \ref{subadditiveTensorProducts}.

    Theorem \ref{thm:THM2} follows from the above since $\epsilon > 0$ was arbitrary and
    $(f_1(\epsilon,n) + f'_2(\epsilon,n))/n \rightarrow 0$ as  $\epsilon \rightarrow 0$. \qed

\section{Discussion}

	We have shown how to build protocols for multiparty distributed compression out of the
	two-party fully quantum Slepian-Wolf protocol. The resulting achievable rates generalize those found
	in \cite{FQSW} for the two-party case and, for the most part, the arguments required are direct
	generalizations of those required for two parties. The most interesting divergence is to be found in section
	\ref{sec:THMIproof}, where we characterize the multiparty rates that can be achieved starting from sequential
	applications of the two-party protocol. These rates are most easily expressed in terms of the
	vertices of the associated polyhedron and we use a graph-theoretic argument to describe the
	polyhedron instead in terms of facet inequalities. We note that it is possible to give a direct proof \cite{HW06}
	that this multiparty rate region is achievable by mimicking the proof techniques of \cite{FQSW}, but in the spirit
	of that paper, we wanted to demonstrate that the more complicated multiparty compression protocols can
	themselves be built out of the simpler near-universal building block of two-party FQSW. Multiparty
	compression thus joins entanglement distillation, entanglement-assisted communication, channel simulation,
	communication over quantum broadcast channels, state redistribution \cite{DY06,O07} and
	many other protocols in the FQSW matriarchy.

	Multiparty FQSW can then itself be used as a building block for other
	multiparty protocols. For example, when classical communication between
	the senders and the receiver is free, combining multiparty FQSW with
	teleportation reproduces the multiparty state merging protocol of
	\cite{HOW05}. Running the protocol backwards in time yields an optimal
	reverse Shannon theorem for broadcast channels \cite{DH07}.

	The achievable rates we describe here, however, are only known to be optimal in the case when
	the source density operator is separable. Otherwise, we proved an outer bound on the rate region
	of the same form as the achievable rate region but with a correction term equal to the
	multiparty squashed entanglement of the source. In order to perform our analysis, we developed
	a number of basic properties of this quantity, notably that it is a convex, subadditive,
	continuous entanglement measure, facts that were established independently in \cite{multisquash}.
	
	We are thus left with some compelling open problems. The most obvious is, of course, to close the gap
	between our inner and
	outer bounds on distributed compression. While that may prove to be
	difficult, some interesting related questions may be easier. For example, can the gap between
	the rate region we have presented here and the true distributed compression region be characterized
	by an entanglement measure? That is, while we have used the multiparty squashed entanglement as a correction
	term, could it be that the true correction term is an entanglement monotone? Also, focusing on the
	squashed entanglement, the two-party version is known to be not just subadditive but additive. Is the same true
	of the multiparty version? \\
	\indent Note added in proof: The additivity of the multiparty squashed entanglement was recently proved in an 
	updated version of \cite{multisquash} which now includes W. Song in the author list.

\section*{Acknowledgments}
	We would like to thank Leonid Chindelevitch, Fr\'ed\'eric Dupuis,
	Michal Horodecki, Jonathan Oppenheim and Andreas Winter for helpful
	comments on the subjects of distributed compression and
	squashed entanglement.
	The authors gratefully acknowledge funding from the Alfred P. Sloan
	Foundation, the Canada Research
	Chairs program, CIFAR, FQRNT, MITACS and NSERC.


\bibliographystyle{unsrt}
\bibliography{esquash}

\end{document}